\documentclass[12pt]{article}
\usepackage{enumerate,amsthm,amsfonts, mathrsfs, amssymb, latexsym, dsfont}
\usepackage{makeidx,setspace, fancyhdr}
\usepackage{amsmath}
\usepackage{avant}
\usepackage{verbatim}
\usepackage{graphicx}
\usepackage{cancel}
\usepackage[left=3cm,right=3cm,top=3cm,bottom=3cm]{geometry}
\usepackage{color}
\usepackage[authoryear]{natbib}
\usepackage{subcaption}
\providecommand{\U}[1]{\protect\rule{.1in}{.1in}}

\newcommand{\BlackBoxes}{\global\overfullrule5pt}

\BlackBoxes
\newcommand{\R}{\mathbb{R}}

\newcommand{\Eop}{\mathbb{E}}
\newcommand{\Pop}{\mathbb{P}}

\newtheorem{theorem}{Theorem}

\newtheorem{lemma}[theorem]{Lemma}

\theoremstyle{definition}
\newtheorem{example}[theorem]{Example}
\newtheorem{remark}[theorem]{Remark}
\newtheorem{definition}[theorem]{Definition}
\numberwithin{equation}{section} 
\numberwithin{theorem}{section}

\begin{document}
\noindent
{\Large\bf Risk management with Tail Quasi-Linear Means}\\

{\bf Nicole B\"auerle${}^{1*}$ and Tomer Shushi${}^{2}$}\\

{${}^{1}$ Institute of Stochastics, Karlsruhe Institute of Technology (KIT), D-76128 Karlsruhe, ${}^{2}$ Department of Business Administration, Guilford Glazer Faculty
of Business and Management, Ben-Gurion University of the Negev, Beer-Sheva, Israel}

${}^{*}$  Corresponding author. E-mail: \emph{nicole.baeuerle@kit.edu}\\

\noindent{\bf Abstract:}

\noindent We generalize Quasi-Linear Means by restricting to the tail of the risk distribution and show that this can be a useful quantity in risk management since it comprises in its general form the Value at Risk, the Conditional Tail Expectation and the Entropic Risk Measure in a unified way. We then investigate the fundamental properties of the proposed measure and show its unique features and implications in the risk measurement process. Furthermore, we derive formulas for truncated elliptical models of losses and provide formulas for selected members of such models.\\

\noindent{ \bf Keywords: }{  Quasi-Linear Means; Risk measurement; Tail risk measures; Conditional Tail Expectation; Value at Risk; Entropic Risk Measure}

\section{Introduction}
One of the most prominent risk measures which are also extensively  used in practice are \emph{Value at Risk} and \emph{Conditional Tail Expectation}. Both have their pros and cons and it is well-known that Conditional Tail Expectation is the smallest coherent (in the sense of \citet{Ar1}) risk measure dominating the Value at Risk (see e.g. \citet{FS}, Theorem 4.67). Though in numerical examples the
Conditional Tail Expectation is often much larger than the Value at Risk, given the same
level $\alpha$. In this paper we present a class of risk measures which
includes both, the Value at Risk and the Conditional Tail Expectation. Another class
with this property is the \emph{Range Value at Risk}, introduced in
\citet{CDS} as a robustification of Value at Risk and Conditional Tail Expectation. Our
approach relies on a generalization of Quasi-Linear Means. Quasi-Linear Means can be traced back to Bonferroni (\citet{b}, p.103) who proposed a unifying formula for different means. Interestingly he motivated this with a problem from actuarial sciences about survival probabilities (for details see also \citet{mp}, p.422).

The Quasi-Linear Mean of a random variable $X$, denoted by $\psi_{U}(X),$ is for an increasing, continuous function $U$ defined as 
\begin{equation*}
\psi_{U}\left(  X\right)  =U^{-1}\left(  \mathbb{E}\left[ U\left(  X\right)
\right]  \right)
\label{eq11}%
\end{equation*}
where $U^{-1}$ is the generalized inverse of $U$ (see e.g. \citet{mp}). If $U$ is in addition concave, $\psi_{U}(X)$ is a {\em Certainty Equivalent}.  If $U$ is convex $\psi_{U}(X)$ corresponds to the {\em Mean Value Risk Measure} (see \citet{HLP}).
We take the actuarial point of view here, i.e. we assume that the random variable
$X$ is real-valued and represents a discounted net loss at the end of a fixed period. This means that positive values are losses whereas negative values are seen as gains.
A well-known risk measure which is obtained when taking the exponential 
function in this definition is the Entropic Risk Measure which is
known to be a convex risk measure but not coherent (see e.g. \citet{M1,T}). 

In this paper, we generalize Quasi-Linear Means by focusing on the tail of the
risk distribution. The proposed measure quantifies the Quasi-Linear Mean of an investor when conditioning on outcomes that are higher than its
Value at Risk. More precisely it is defined by
\[
\rho_{U}^{\alpha}(X):=U^{-1}\Big(\mathbb{E}\big[U(X)|X\geq VaR_{\alpha
}(X)\big]\Big).
\]
where $VaR_{\alpha}$ is the usual Value at Risk. We call it Tail Quasi-Linear Mean (TQLM). It can be shown that when we restrict to concave (utility)
 functions, the TQLM interpolates between the Value at Risk and the
Conditional Tail Expectation. By choosing the utility function $U$ in the right way we
can either be close to Value at Risk or the Conditional Tail Expectation. Both extreme
cases are also included when we plug in specific utility functions. The Entropic Risk Measure is also a limiting case of our construction. Though in
general not being convex, the TQLM has some nice properties. In particular it
is still manageable and useful in applications. We show the application of
TQLM risk measures for capital allocation, for optimal reinsurance
and for finding minimal risk portfolios. In particular within the class of
symmetric distributions we show that explicit computations lead to analytic closed-forms of TQLM.

In the actuarial sciences there are already some approaches to unify risk measures or premium principles. Risk measures can be seen as a broader concept than insurance premium principles since the latter one is considered as a ''price'' of a risk (for a discussion see e.g.\ \citet{GKD,FZ}). Both are in its basic definition mappings from the space of random variables into the real numbers, but the interesting properties may vary with the application. In \citet{GKD} a unifying approach to derive risk measures and premium principles has been proposed by minimizing a Markov bound for the tail probability. The approach includes among others the Mean Value principle, the Swiss premium principle and Conditional Tail Expectation.

In \citet{FZ} weighted premiums have been introduced where the expectation is taken with respect to a weighted distribution function. This construction includes e.g.\ the Conditional Tail Expectation, the Tail Variance and the Esscher premium. This paper also discusses invariance and additivity properties of these measures.

Further, the Mean Value Principle has been generalized in various ways. In \citet{BGGS} these premium principles have been
extended to the so-called \emph{Swiss Premium Principle} which interpolates
with the help of a parameter $z\in[0,1]$ between the Mean Value Principle and
the Zero-Utility Principle. Properties of the Swiss Premium Principle have
been discussed in \citet{VG}. In particular monotonicity, positive
subtranslativity and subadditivity for independent random variables are shown
under some assumptions. The latter two notions are weakened versions of
translation invariance and subadditivity, respectively.

The so-called {\em Optimized Certainty Equivalent} has been investigated in \citet{B1} as a mean to construct risk measures. It comprises the Conditional Tail Expectation and  bounded shortfall risk.

The following Section provides definitions and preliminaries on risk measures
that will serve as necessary foundations for the paper. Section 3 introduces
the proposed risk measure and derives its fundamental properties. We show
various representations of this class of risk measures and prove for concave
 functions $U$  (under a technical assumption) that the TQLM is bounded
between the Value at Risk and the Conditional Tail Expectation. Unfortunately the only
coherent risk measure in this class turns out to be the Conditional Tail Expectation
(this is maybe not surprising since this is also true within the class of
ordinary Certainty Equivalents, see \citet{M1}). In Section 4 we consider the
special case when we choose the exponential function.  In this case we call $\rho_U^\alpha$ Tail Conditional Entropic Risk Measure and show that it  is convex within the class of comonotone random variables. Section 5 is devoted to applications. In the
first part we discuss the application to capital allocation. We
define a risk measure for each subportfolio based on our TQLM and discuss its properties. In the second part we
consider an optimal reinsurance problem with the TQLM as target function. For
convex functions $U$ we show that the optimal reinsurance treaty is of
stop-loss form. In Section 6, the proposed risk measure is investigated for
the family of symmetric distributions. Some explicit calculations can be done
there. In particular there exists an explicit formula for the Tail Conditional
Entropic Risk Measure. Finally a minimal risk portfolio problem is solved when
we consider the Tail Conditional Entropic Risk Measure as target function.
Section 7 offers a discussion to the paper.

\section{Classical risk measures and other preliminaries}
We consider real-valued continuous random variables $X: \Omega\to\mathbb{R}$ defined on a probability space $(\Omega, \mathcal{F},\mathbb{P})$ and denote this set by $\mathcal{X}$. These random
variables represent discounted net losses at the end of a fixed period, i.e.\ positive values are seen as losses whereas negative values are seen as gains. We
denote the (cumulative) distribution function by $F_{X}(x):=\mathbb{P}(X \leq
x), x\in\mathbb{R}$. 
Moreover we consider increasing and continuous functions $U:\R\to\R$ (in case $X$ takes only positive or negative values, the domain of $U$ can be restricted). The generalized inverse $U^{-1}$ of such a function is defined by
\[
U^{-1}(x):=\inf\{y\in\mathbb{R}:U(y)\geq x\},
\]
where $x \in\mathbb{R}$. With
\[
L^{1}:=\{X \in \mathcal{X}: \ X \text{ is a random variable with }
\mathbb{E}[X] < \infty\}
\]
we denote the space of all real-valued, continuous, integrable random variables. We now
recall some notions of risk measures. In general, a risk measure is a mapping
$\rho: L^{1} \to\mathbb{R}\cup\{\infty\}$. Of particular importance are the
following risk measures.

\begin{definition}
For $\alpha\in(0,1)$ and $X \in L^{1}$ with distribution function $F_{X}$ we define
\begin{itemize}
\item[a)] the \emph{Value at Risk} of $X$ at level $\alpha$ as $VaR_{\alpha
}(X) := \inf\{x\in\mathbb{R}: F_{X}(x)\geq\alpha\}$.

\item[b)] the \emph{Conditional Tail Expectation} of $X$ at level $\alpha$ as
\[
CTE_{\alpha}(X):= \mathbb{E}[ X | X \ge VaR_{\alpha}(X)].
\]
\end{itemize}
\end{definition}

Note that the definition of Conditional Tail Expectation is for continuous random variables the same as the Average Value at Risk, the Expected Shortfall  or the  Tail Conditional Expectation (see chapter 4 of \citet{FS} or \citet{DDGK}). Below we summarize some properties of the generalized inverse (see e.g.
\citet{mfe}, A.1.2).

\begin{lemma}
\label{lem:gi} For an increasing, continuous function $U$ with generalized inverse $U^{-1}$ it holds:

\begin{itemize}
\item[a)] $U^{-1}$ is strictly increasing and left-continuous.
\item[b)] For all $x\in\mathbb{R}_{+}, y\in\mathbb{R}$, we have $U^{-1} \circ
U(x)\le x$ and $U\circ U^{-1}(y) = y.$
\item[c)] If $U$ is strictly increasing on $(x-\varepsilon,x)$ for an
$\varepsilon>0$, we have $U^{-1} \circ U(x)= x$.
\end{itemize}
\end{lemma}

The next lemma is useful for alternative representations of our risk measure.
It can be directly derived from the definition of Value at Risk.

\begin{lemma}
\label{lem:VarU} For $\alpha\in(0,1)$ and any increasing, left-continuous
function $f: \mathbb{R} \to\mathbb{R}$ it holds $VaR_{\alpha}(f(X)) =
f\big(VaR_{\alpha}(X)\big).$
\end{lemma}

In what follows we will study some properties of risk measures $\rho: L^{1}%
\to\mathbb{R}\cup\{\infty\}$, like

\begin{itemize}
\item[(i)] \emph{law-invariance:} $\rho(X)$ depends only on the distribution
$F_{X}$.
\item[(ii)] \emph{constancy:} $\rho(m)=m$ for all $m\in\mathbb{R}_{+}$.
\item[(iii)] \emph{monotonicity:} If $X\le Y$ then $\rho(X)\le\rho(Y)$.
\item[(iv)] \emph{translation invariance:} For $m\in\mathbb{R}$ it holds
$\rho(X+m)=\rho(X)+m$.
\item[(v)] \emph{positive homogeneity:} For $\lambda\ge0$ it holds that
$\rho(\lambda X)=\lambda\rho(X).$
\item[(vi)] \emph{subadditivity:} $\rho(X+Y)\le\rho(X)+\rho(Y)$.
\item[(vii)] \emph{convexity:} For $\lambda\in[0,1]$ it holds that
$\rho(\lambda X+(1-\lambda)Y)\le\lambda\rho(X)+(1-\lambda)\rho(Y).$
\end{itemize}

A risk measure with the properties (iii)-(vi) is called \emph{coherent}. Note
that $CTE_{\alpha}(X)$ is not necessarily coherent when $X$ is a
discrete random variable, but is coherent if $X$ is continuous. Also note that
if $\rho$ is positive homogeneous, then convexity and subadditivity are
equivalent properties. Next we need the notion of the usual stochastic
ordering (see e.g. \citet{ms02}).

\begin{definition}
Let $X,Y$ be two random variables. Then $X$ is less than $Y$ in \emph{usual
stochastic order} ($X \leq_{st}Y$) if $\mathbb{E}[f(X)] \leq\mathbb{E}[f(Y)]$
for all increasing $f:\mathbb{R}\to\mathbb{R}$, whenever the expectations
exist. This is equivalent to $F_{X}(t) \ge F_{Y}(t)$ for all $t\in\mathbb{R}$.
\end{definition}

Finally we also have to deal with comonotone random variables (see e.g.\ Definition 1.9.1 in \citet{DDGK});

\begin{definition}
Two random variables $X,Y$ are called \emph{comonotone} if there exists a
random variable $Z$ and increasing functions $f,g:\mathbb{R}\to\mathbb{R}$
such that $X=f(Z)$ and $Y=g(Z)$. The pair is called \emph{countermonotone} if one of the two functions is increasing, the other decreasing.
\end{definition}

\section{Tail Quasi-Linear Means}
For continuous random variables $X\in \mathcal{X}$ and levels $\alpha\in(0,1)$ let us introduce risk measures of the following form:

\begin{definition}
Let $X\in \mathcal{X}$, $\alpha\in (0,1)$ and $U$ an increasing, continuous  function. The \emph{Tail Quasi-Linear Mean}  is defined by
\begin{equation}
\rho_{U}^{\alpha}(X):=U^{-1}\Big(\mathbb{E}\big[U(X)|X\geq VaR_{\alpha
}(X)\big]\Big)
\end{equation}
whenever the conditional expectation inside exists and is finite.
\end{definition}

\begin{remark}
\begin{itemize}
\item[a)]  It is easy to see that $U(x)=x$ leads to $CTE_{\alpha}(X)$.
\item[b)] The Quasi-Linear Mean $\psi_U(X)$ is obtained by taking $\lim_{\alpha\downarrow 0} \rho_U^\alpha(X).$
\end{itemize}
\end{remark}

 In what follows we will first give some alternative representations of the TQLM. By definition of the conditional distribution it follows immediately that we can write
\[
\rho_{U}^{\alpha}(X)= U^{-1}\left(  \frac{\mathbb{E}\big[U(X) 1_{\{X\geq
VaR_{\alpha}(X)\}} \big]}{\mathbb{P}(X\ge VaR_{\alpha}(X))}\right)
\]
where $\mathbb{P}(X\ge VaR_{\alpha}(X))=1-\alpha$ for  continuous $X$.
Moreover, when we denote by $\tilde{\mathbb{P}}(\cdot)= \mathbb{P}(\cdot|
X\geq VaR_{\alpha}(X))$ the conditional probability given $X\geq VaR_{\alpha
}(X)$, then we obtain
\begin{equation}
\label{eq:condexp}\rho_{U}^{\alpha}(X)=U^{-1}\Big(\tilde{\mathbb{E}%
}\big[U(X)\big]\Big).
\end{equation}
Thus, $\rho_{U}^{\alpha}(X)$ is just the Quasi-Linear Mean of $X$ with respect to the
conditional distribution. In order to get an idea what the TQLM measures, suppose that $U$ is sufficiently differentiable. Then we get by a Taylor series approximation (see e.g. \citet{bp}) that
\begin{equation}\label{eq:TQLMapprox}
\rho_U^\alpha(X) \approx CTE_\alpha(X)-\frac12 \ell_U(CTE_\alpha(X)) TV_\alpha(X)
\end{equation}
with $\ell_U(x) = -\frac{U''(x)}{U'(x)}$ being the Arrow-Pratt function of absolute risk aversion  and 
\begin{equation}\label{eq:tailvar}
TV_{\alpha}\left(  X\right)  :=Var\left(  X|X\geq VaR_{\alpha}\left(X\right)  \right)  = 
\Eop[(X-CTE_{\alpha}(  X))^2 | X> VaR_\alpha(X)] 
\end{equation}
 being the tail variance of $X.$ If $U$ is concave $\ell_U \ge 0$ and $TV_\alpha$ is subtracted from $CTE_\alpha$, if $U$ is convex $\ell_U \le 0$ and $TV_\alpha$ is added, penalizing deviations in the tail. In this sense $\rho_U^\alpha(X) $ is approximately a Lagrange-function of a restricted optimization problem where we want to optimize the Conditional Tail Expectation under the restriction that the tail variance is not too high. 

The following technical assumption will be useful:

\begin{description}
\item[(A)] There exists an $\varepsilon>0$ such that $U$ is strictly increasing
on $( VaR_{\alpha}(X) -\varepsilon,VaR_{\alpha}(X) )$.
\end{description}

Obviously assumption (A) is satisfied if $U$ is strictly increasing on its domain which should be satisfied in all reasonable applications. Economically (A) states that at least shortly before the critical level $VaR_{\alpha}(X)$ our measure  strictly penalizes higher outcomes of $X$.  Under assumption (A) we obtain another representation of the TQLM.

\begin{lemma}
For all $X\in \mathcal{X}$, increasing continuous functions $U$ and $\alpha\in(0,1)$ such that (A) is satisfied  we have that
\[
\rho_{U}^{\alpha}(X)=U^{-1}\Big( CTE_{\alpha}(U(X))\Big).
\]
\end{lemma}

\begin{proof}
We first show that under (A) we obtain:
\[
\{ X\ge VaR_{\alpha}(X)\} = \{ U(X)\ge VaR_{\alpha}(U(X))\}.
\]
Due to Lemma \ref{lem:VarU} we immediately obtain
\[
\{ X\ge VaR_{\alpha}(X)\} \subset\{ U(X)\ge U(VaR_{\alpha}(X))\}=\{ U(X)\ge
VaR_{\alpha}(U(X))\}.
\]
On the other hand we have with Lemma \ref{lem:gi} b),c) that
\[
U(X)\ge VaR_{\alpha}(U(X)) \Rightarrow X\ge U^{-1} \circ U(X)\ge U^{-1} \circ
U(VaR_{\alpha}(X))=VaR_{\alpha}(X)
\]
which implies that both sets are equal.

Thus, we get that
\[
\mathbb{E}\big[U(X)|X\geq VaR_{\alpha}(X)\big] = \mathbb{E}\big[U(X)|U(X)\geq
VaR_{\alpha}(U(X))\big] = CTE_{\alpha}(U(X))
\]
which implies the statement.
\end{proof}

Next we provide some simple yet fundamental features of the TQLM. The first one is rather obvious and we skip the proof.

\begin{lemma}
For any $ X \in \mathcal{X}$, the TQLM and the Quasi-Linear Mean $\psi_{U}$ are related as follows:
\[
\rho_{U}^{\alpha}(X) \ge\psi_{U}(X).
\]
\end{lemma}

The TQLM interpolates between the Value at Risk and the Conditional Tail Expectation in case $U$ is concave. We will show this in the next theorem under our
assumption (A) (see also Figure \ref{fig:obj}):

\begin{theorem}
For $X\in \mathcal{X}$ and concave increasing functions $U$ and
$\alpha\in(0,1)$  such that (A) is satisfied we have that
\[
VaR_{\alpha}(X) \le\rho_{U}^{\alpha}(X) \le CTE_{\alpha}(X).
\]
Moreover, there exist utility functions such that the bounds are attained. In
case $U$ is convex and satisfies (A) and all expectations exist, we obtain
\[
\rho_{U}^{\alpha}(X) \ge CTE_{\alpha}(X).
\]
\end{theorem}

\begin{proof}
Let $U$ be concave. We will first prove the upper bound. Here we use the
representation of $\rho_{U}^{\alpha}(X)$ in \eqref{eq:condexp} as a Certainty
Equivalent of the conditional distribution $\tilde{\mathbb{P}}$. We obtain
with the Jensen inequality
\begin{equation}
\label{eq:Jensen}\tilde{\mathbb{E}}[U(X)] \le U(\tilde{\mathbb{E}}[X]) =
U(CTE_{\alpha}(X)).
\end{equation}
Taking the generalized inverse of $U$ on both sides and using Lemma
\ref{lem:gi} a), b) yields
\[
\rho_{U}^{\alpha}(X) \le U^{-1}\circ U(CTE_{\alpha}(X))\le CTE_{\alpha}(X).
\]
The choice $U(x)=x$ leads to $\rho_{U}^{\alpha}(X) =CTE_{\alpha}(X).$

For the lower bound first note that
\[
U(VaR_{\alpha}(X))\leq\mathbb{E}\big[U(X)|X\geq VaR_{\alpha}(X)\big].
\]
Taking the generalized inverse of $U$ on both sides and using Lemma
\ref{lem:gi} c) yields
\[
VaR_{\alpha}(X)=U^{-1}\circ U(VaR_{\alpha}(X))\leq\rho_{U}^{\alpha}(X).
\]
Defining
\[
U(x)=\left\{
\begin{array}
[c]{cl}%
x, & x\leq VaR_{\alpha}(X)\\
VaR_{\alpha}(X), & x>VaR_{\alpha}(X)
\end{array}
\right.
\]
yields
\[
U^{-1}(x)=\left\{
\begin{array}
[c]{cl}%
x, & x\leq VaR_{\alpha}(X)\\
\infty, & x>VaR_{\alpha}(X)
\end{array}
\right.
\]
and we obtain
\[
\mathbb{E}\big[U(X)|X\geq VaR_{\alpha}(X)\big]=U(VaR_{\alpha}(X)).
\]
Taking the generalized inverse of $U$ on both sides and using Lemma
\ref{lem:gi} c) yields
\[
\rho_{U}^{\alpha}(X)=U^{-1}\circ U(VaR_{\alpha}(X))=VaR_{\alpha}(X)
\]
which shows that the lower bound can be attained. If $U$ is convex, the
inequality in \eqref{eq:Jensen} reverses.
\end{proof}

\begin{figure}[h]
\centering
\includegraphics[width=0.95\textwidth]{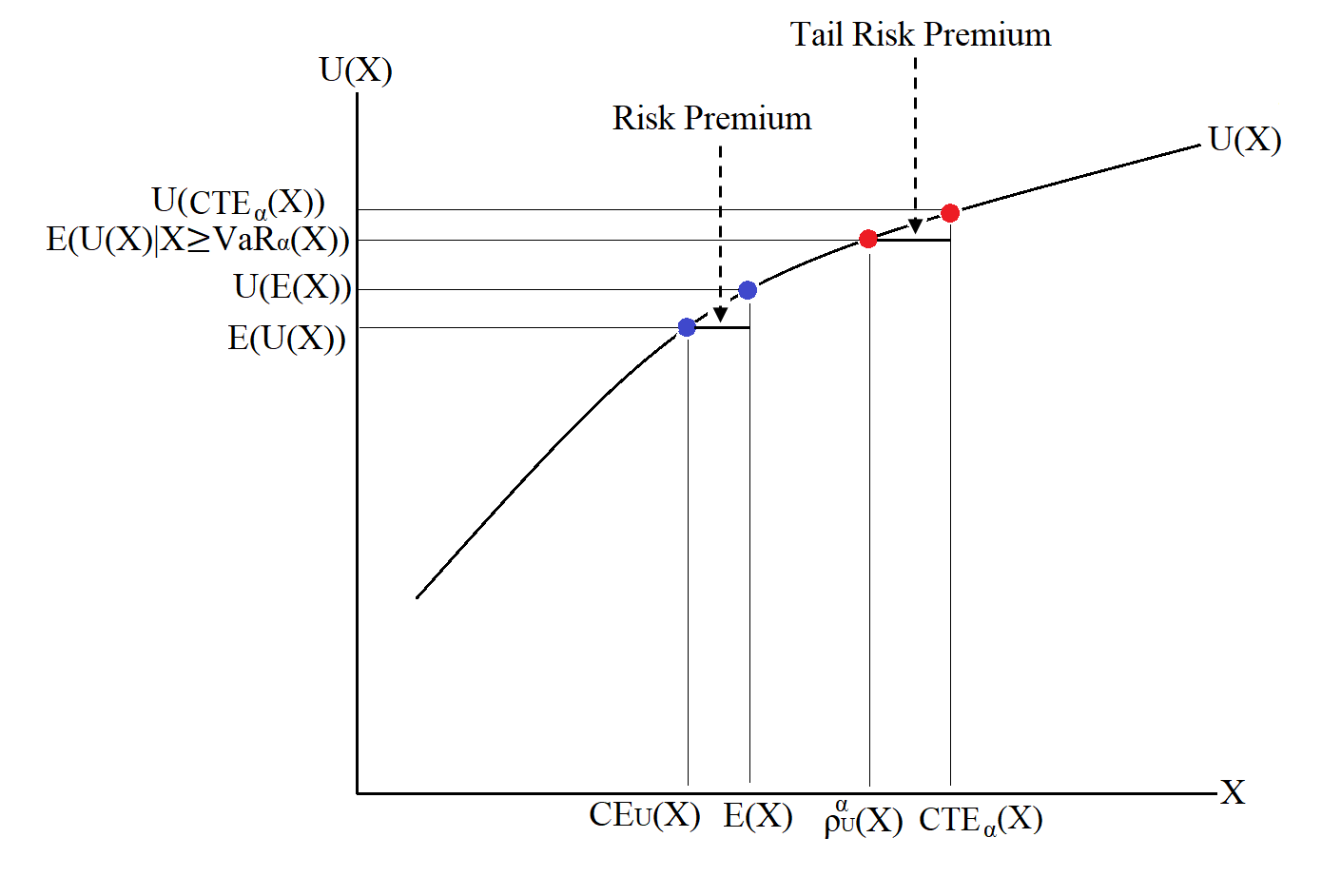} \caption{Relation between
the TQLM, the CTE, the Certainty Equivalent and the expectation in case the
utility function $U$ is concave.}%
\label{fig:obj}%
\end{figure}

Next we discuss properties of the TQLM. Of course when we choose $U$ in a specific way we expect more properties to hold. 

\begin{theorem}
\label{theo:TCCEprop} The TQLM $\rho_{U}^{\alpha}$ has the following properties:
\begin{itemize}
\item[a)] It is law-invariant.
\item[b)] It has the constancy property.
\item[c)] It is monotone.
\item[d)] It is translation-invariant within the class of functions
which are strictly increasing if and only if $U(x)=-e^{-\gamma x},\gamma
>0,$\ or if $U$ is linear.
\item[e)] It is positive homogeneous within the class of functions
which are strictly increasing if and only $U(x)=\frac{1}{\gamma}x^{\gamma
},x>0,\gamma\neq0$ or $U(x)=\ln(x)$ or $U$ is linear.
\end{itemize}
\end{theorem}

\begin{proof}
\begin{itemize}
\item[a)] The law-invariance follows directly from the definition of $\rho
_{U}^{\alpha}$ and the fact that $VaR_{\alpha}$ is law-invariant.
\item[b)] For $m\in\mathbb{R}$ we have that $VaR_{\alpha}(m)=m$ and thus
$\tilde{\mathbb{P}}=\mathbb{P}$ which implies the statement.
\item[c)] We use here the representation
\[
\rho_{U}^{\alpha}(X)= U^{-1}\left(  \frac{\mathbb{E}\big[U(X) 1_{\{X\geq
VaR_{\alpha}(X)\}} \big]}{1-\alpha}\right).
\]
Thus it suffices to show that the relation  $X\le Y$ implies $\mathbb{E}\big[U(X) 1_{\{X\geq
VaR_{\alpha}(X)\}} \big]\le \mathbb{E}\big[U(Y) 1_{\{Y\geq
VaR_{\alpha}(Y)\}} \big].$   Since we are only interested in the marginal distributions of $X$ and $Y$ we can choose $X=F_{X}^{-1}(V),Y=F_{Y}^{-1}(V)$ with same random variable $V$ which is uniformly distributed on $(0,1)$. We obtain with Lemma \ref{lem:gi}
\begin{align*}
&  X\geq VaR_{\alpha}(X)\Leftrightarrow F_{X}^{-1}(V)\geq VaR_{\alpha}
(F_{X}^{-1}(V))\Leftrightarrow F_{X}^{-1}(V)\geq F_{X}^{-1}\big(VaR_{\alpha
}(V)\big)\\
&  \Leftrightarrow F_{X}^{-1}(V)\geq F_{X}^{-1}\big(\alpha\big)\Leftrightarrow
V\geq\alpha.
\end{align*}
The same holds true for $Y$. Since $X\le Y$ we obtain $F_X^{-1} \le F_Y^{-1}$ and thus
\begin{eqnarray*}
\mathbb{E}\big[U(X) 1_{\{X\geq VaR_{\alpha}(X)\}} \big] &=& \mathbb{E}\big[F_X^{-1}(V) 1_{\{V\geq \alpha\}} \big] \\ &\le&  \mathbb{E}\big[F_Y^{-1}(V) 1_{\{V\geq \alpha\}} \big] = \mathbb{E}\big[U(Y) 1_{\{Y\geq VaR_{\alpha}(Y)\}} \big]
\end{eqnarray*}
which implies the result.
\item[d)] Since we have the representation
\begin{equation}
\rho_{U}^{\alpha}(X)=U^{-1}\Big(\tilde{\mathbb{E}}\big[U(X)\big]\Big),
\end{equation}
this statement follows from \citet{M1}, Theorem 2.2. Note that we can work here
with one fixed conditional distribution since $\{X\ge VaR_{\alpha}(X)\}=\{X+c
\ge VaR_{\alpha}(X+c)\}$ for all $c\in\mathbb{R}$.
\item[e)] As in d) this statement follows from \citet{M1}, Theorem 2.3. Note
that we can work here with one fixed conditional distribution since $\{X\ge
VaR_{\alpha}(X)\}=\{\lambda X \ge VaR_{\alpha}(\lambda X)\}$ for all
$\lambda>0$.
\end{itemize}
\end{proof}

\begin{remark}
The monotonicity property of Theorem \ref{theo:TCCEprop} seems to be obvious,
but it indeed may not hold if $X$ and $Y$ are discrete. 
One has to be cautious in this case (see also the examples given in \citet{bm06}). The same is true for the Conditional Tail Expectation.
\end{remark}

\begin{theorem}
If $\rho_{U}^{\alpha}$ is a coherent risk measure, then it is the Conditional Tail Expectation Measure $\rho_{U}^{\alpha
}(X)=CTE_{\alpha}\left(  X\right)  .$
\end{theorem}

\begin{proof}
As can be seen from Theorem \ref{theo:TCCEprop}, the translation invariance
and homogeneity properties hold simultaneously if and only $U$ is linear,
which implies that $\rho_{U}^{\alpha}$ is the Conditional Tail Expectation.
\end{proof}

\section{Tail Conditional Entropic Risk Measure}
In case $U(x)=\frac1\gamma e^{\gamma x},\gamma\neq0,$ we obtain a conditional tail version of the Entropic Risk Measure. It is given by
\begin{equation}\label{eq:TCERM}
\rho_{U}^{\alpha}(X)=\frac{1}{\gamma}\log\mathbb{E}[e^{\gamma X}|X\geq
VaR_{\alpha}(X)].
\end{equation}
In this case we write $\rho_{\gamma}^{\alpha}$ instead of $\rho_{U}^{\alpha}$ since $U$ is determined by $\gamma$. For $\alpha\downarrow 0$ we obtain in the limit the classical Entropic Risk Measure.
We call $\rho_{\gamma}^{\alpha}(X)$ {\em Tail Conditional Entropic Risk Measure} and get from \eqref{eq:TQLMapprox} the following approximation of $\rho_{\gamma}^{\alpha}(X)$:
If $\gamma\neq0$ is sufficiently close to zero, the conditional tail version
of the Entropic Risk Measure can be approximated by
\[
\rho_{\gamma}^{\alpha}(X)\approx CTE_{\alpha}\left(  X\right)  -\frac{\gamma
}{2}TV_{\alpha}\left(  X\right).
\]
i.e. it is a weighted measure consisting of Conditional Tail Expectation and Tail Variance (see \eqref{eq:tailvar}).

Another representation of the Tail Conditional Entropic Risk Measure is for $\gamma\neq0$ given by (see e.g. \citet{BR,B1})
\[
\rho_{\gamma}^{\alpha}(X)=\inf_{\mathbb{Q}\ll\tilde{\mathbb{P}}}\left(
\mathbb{E}_{\mathbb{Q}}[X]+\frac{1}{\gamma}\mathbb{E}_{\mathbb{Q}}\left(
\log\frac{d\mathbb{Q}}{d\tilde{\mathbb{P}}}\right)  \right)  .
\]
where $\tilde{\mathbb{P}}$ is again the conditional distribution
$\mathbb{P}(\cdot|X\geq VaR_{\alpha}(X))$. The minimizing ${\mathbb{Q}^{\ast}}$
is attained at
\[
\mathbb{Q}^{\ast}(dz)=\frac{e^{\gamma z}\tilde{\mathbb{P}}(dz)}{\int e^{\gamma
y}\tilde{\mathbb{P}}(dy)}.
\]

According to Theorem \ref{theo:TCCEprop} we cannot expect the Tail Conditional
Entropic Risk Measure to be convex. However we obtain the following result:

\begin{theorem}
\label{theo:TCEconvex} For $\gamma>0$ the Tail Conditional Entropic Risk
Measure is convex for comonotone  random variables.
\end{theorem}

\begin{proof}
First note that the Tail Conditional Entropic Risk Measure has the constancy
property and is translation invariant. Thus, using
Theorem 6 in \citet{DG} it is sufficient to show that $g^{\prime\prime
}(0;X,Y)\geq0$ for all comonotone $X,Y$ where
\[
g(t;X,Y)=\rho_{\gamma}^{\alpha}(X+t(Y-X)),\quad t\in(0,1).
\]
Since $X$ and $Y$ are comonotone we can write them as $X=F_{X}^{-1}
(V),Y=F_{Y}^{-1}(V)$ with same random variable $V$ which is uniformly
distributed on $(0,1)$. Thus we get with Lemma \ref{lem:gi} (compare also the proof of Theorem \ref{theo:TCCEprop} c))
\begin{align*}
&  X\geq VaR_{\alpha}(X)\Leftrightarrow F_{X}^{-1}(V)\geq VaR_{\alpha}
(F_{X}^{-1}(V))\Leftrightarrow F_{X}^{-1}(V)\geq F_{X}^{-1}\big(VaR_{\alpha
}(V)\big)\\
&  \Leftrightarrow F_{X}^{-1}(V)\geq F_{X}^{-1}\big(\alpha\big)\Leftrightarrow
V\geq\alpha.
\end{align*}
The same holds true for $Y$ and also for $X+t(Y-X)=(1-t)X+tY=(1-t)F_{X}
^{-1}(V)+tF_{Y}^{-1}(V)$ since it is an increasing, left-continuous function of $V$ for
$t\in(0,1)$. Thus all events on which we condition here are the same:
\[
\{X\geq VaR_{\alpha}(X)\}\!=\!\{Y\geq VaR_{\alpha}(Y)\}\!=\!\{X+t(Y-X)\geq
VaR_{\alpha}(X+t(Y-X))\}\!=\!\{V\geq\alpha\}.
\]
Hence we obtain
\[
g^{\prime}(t;X,Y)=\frac{\mathbb{E}\left[  (Y-X)e^{\gamma(X+t(Y-X))}
1_{[V\geq\alpha]}\right]  }{\mathbb{E}\left[  e^{\gamma(X+t(Y-X))}
1_{[V\geq\alpha]}\right]  }
\]
and
\[
g^{\prime\prime}(0;X,Y)=\gamma\left\{  \frac{\mathbb{E}\left[  (Y-X)^{2}
e^{\gamma X}1_{[V\geq\alpha]}\right]  }{\mathbb{E}\left[  e^{\gamma
X}1_{[V\geq\alpha]}\right]  }-\left(  \frac{\mathbb{E}\left[  (Y-X)e^{\gamma
X}1_{[V\geq\alpha]}\right]  }{\mathbb{E}\left[  e^{\gamma X}1_{[V\geq\alpha
]}\right]  }\right)  ^{2}\right\}  .
\]
This expression can be interpreted as the variance of $(Y-X)$ under the
probability measure
\[
\frac{d\mathbb{P}^{\prime}}{d\mathbb{P}}=\frac{e^{\gamma X}1_{[V\geq\alpha]}
}{\mathbb{E}\left[  e^{\gamma X}1_{[V\geq\alpha]}\right]  }
\]
and is thus greater or equal to zero which implies the statement.
\end{proof}

\section{Applications}
In this Section we show that the TQLM is a useful tool for various applications in risk management.

\subsection{Capital Allocation}
Firms often have the problem of allocating a global risk capital requirement down to subportfolios. One way to do this is to use Aumann-Shapley capital allocation rules. For convex risk measures this is not an easy task and has e.g.\  been discussed in \citet{T}.
A desirable property in this respect would be that the sum of the capital requirements for the subportfolios equals the global risk capital requirement. More precisely, let $\left(  X_{1},X_{2},...,X_{n}\right)  $ be a vector of
$n$ random variables  and let  $S=X_{1}+X_{2}+...+X_{n}$ be its sum. An intuitive way to measure the contribution of $X_i$ to the total capital requirement, based on the TQLM is by defining (compare for instance with  \citet{L1}):
\[
\rho_{U}^{\alpha}\left(  X_{i}|S\right)  :=U^{-1}\left(  \mathbb{E}\left[
U\left(  X_{i}\right)  |S\geq VaR_{\alpha}\left(  S\right)  \right] \right).
\]
This results in a capital allocation rule if
\begin{equation}
\rho_{U}^{\alpha}\left(  S\right)  =\sum_{i=1}^{n}\rho_{U}^{\alpha}\left(
X_{i}|S\right)  . \label{r_S1}%
\end{equation}
It is easily shown that this property is only true in a special case:

\begin{theorem}
The TQLM of the aggregated loss $S$ is equal to the sum of TQLM of the risk
sources $X_{i},i=1,2,...,n,$ i.e. \eqref{r_S1} holds for all random variables $X_{i},i=1,2,...,n,$ if and only if $U$ is linear.
\end{theorem}

In general we cannot expect \eqref{r_S1} to hold. Indeed for the Tail Conditional Entropic Risk Measure we obtain that in case the losses are comonotone, it is not profitable to split the portfolio in subportfolios, whereas it is profitable if two losses are countermonotone.

\begin{theorem}
The Tail Conditional Entropic Risk Measure has for $\gamma>0$ and comonotone $X_i, i=1,\ldots, n$ the property that
\[
\rho^{\alpha}_{\gamma}(S) \ge\sum_{i=1}^{n} \rho_{\gamma}^{\alpha}(X_{i}|S).
\]
In case $n=2$ and $X_1, X_2$ are countermonotone the inequality reverses.
\end{theorem}

\begin{proof}
As  in the proof of Theorem \ref{theo:TCEconvex} we get for  comonotone  $X,Y$ that   $X=F_{X}^{-1}(V), Y=F_{Y}^{-1}(V)$ with same random variable $V$ which is uniformly distributed on $(0,1)$ and that 
\begin{align*}
&  X+Y\geq VaR_{\alpha}(X+Y)\Leftrightarrow  V\geq\alpha.
\end{align*}
Thus with $S=X+Y$
\begin{eqnarray*}
\frac{1}{1-\alpha}\Eop\Big[e^{\gamma(X+Y)} 1_{[S\ge VaR_\alpha(S)]}\Big] &=& \frac{1}{1-\alpha} \Eop\Big[e^{\gamma (F_X^{-1}(V)+F_Y^{-1}(V))} 1_{[V\ge \alpha]}   \Big] \\
&=& \tilde{\Eop} \Big[ e^{\gamma F_X^{-1}(V)} e^{\gamma F_Y^{-1}(V)}  \Big]\\
&\ge &  \tilde{\Eop}\Big[e^{\gamma F_X^{-1}(V)}  \Big]  \tilde{\Eop}\Big[ e^{\gamma F_Y^{-1}(V)}   \Big] \\
&=&  \frac{1}{1-\alpha}\Eop\Big[e^{\gamma X} 1_{[S\ge VaR_\alpha(S)]}\Big]  \frac{1}{1-\alpha}\Eop\Big[e^{\gamma Y} 1_{[S\ge VaR_\alpha(S)]}\Big] 
\end{eqnarray*}
since the covariance is positive for comonotone random variables. Here, as before $\tilde{\Pop}$ is the conditional distribution given by $\frac{d\tilde{\Pop}}{d\Pop}=\frac{1}{1-\alpha} 1_{[V\ge \alpha]}$. Taking $\frac{1}{\gamma} \log$ on both sides implies the result for $n=2$. The general result follows by induction on the number of random variables. In the countermonotone case the inequality reverses.
\end{proof}

\subsection{Optimal Reinsurance}
TQLM risk measures can also be used to find optimal reinsurance treaties. In
case the random variable $X$ describes a loss, an insurance company is able to
reduce its risk by splitting $X$ into two parts and transferring one of these
parts to a reinsurance company. More formally a reinsurance treaty is a
function $f:\mathbb{R}_{+}\to\mathbb{R}_{+}$ such that $f(x)\le x$ and $f$ as
well as $R_{f}(x) := x-f(x)$ are both increasing. The reinsured part is then
$f(x)$. The latter assumption is often made to rule out moral hazard. In what
follows let
\[
\mathcal{C} = \{f:\mathbb{R}_{+} \to\mathbb{R}_{+} | \ f(x) \leq x \ \forall x
\in\mathbb{R}_{+} \text{ and } f, R_{f} \text{ are increasing}\},
\]
be the set of all reinsurance treaties. Note that functions in $\mathcal{C}$
are in particular Lipschitz-continuous, since $R_{f}$ increasing leads to
$f(x_{2})-f(x_{1}) \le x_{2}-x_{1}$ for all $0\le x_{1}\le x_{2}$. Of course
the insurance company has to pay a premium to the reinsurer for taking part of
the risk. For simplicity we assume here that the premium is computed according
to the expected value premium principle, i.e. it is given by $(1+\theta)
\mathbb{E}[f(X)]$ for $\theta>0$ and a certain amount $P>0$ is available for
reinsurance. The aim is now to solve
\begin{equation}
\label{eq:ORP}\min_f\quad\rho_{U}^{\alpha}\big( R_{f}(X) \big) \quad
s.t.\quad(1+\theta) \mathbb{E}[f(X)] =P,\; f \in\mathcal{C}.
\end{equation}
This means that the insurance company tries to minimize the retained risk,
given the amount $P$ is available for reinsurance. Problems like this can e.g.
be found in \citet{GZ}. A multidimensional extension is given in \citet{bg18}.
In what follows we assume that $U$ is strictly increasing, strictly convex and
continuously differentiable, i.e. according to \eqref{eq:TQLMapprox} large deviations in the right tail of $R_f(X)$ are heavily penalized. In order to avoid trivial cases we assume that the available amount of money for reinsurance is not too high, i.e. we assume that
\[
P< (1+\theta)\mathbb{E}[ (X-VaR_{\alpha}(X))_{+}].
\]

The optimal reinsurance treaty is given in the following theorem. It turns out to be a stop-loss treaty.

\begin{theorem}
The optimal reinsurance treaty of problem \eqref{eq:ORP} is given by
\[
f^{*}(x) = \left\{
\begin{array}
[c]{cl}%
0, & x \le a,\\
x-a, & x> a
\end{array}
, \right.
\]
where $a$ is a positive solution of $(1+\theta)\mathbb{E}[ (X-a)_{+}] =P.$
\end{theorem}

Note that the optimal reinsurance treaty does not depend on the precise form of $U$, i.e. on the precise risk aversion of the insurance company.

\begin{proof}
First observe that $z \mapsto\mathbb{E}[ (X-z)_{+}]$ is continuous and
decreasing. Moreover by assumption $P< (1+\theta)\mathbb{E}[ (X-VaR_{\alpha}(X))_{+}].$
Thus by the mean-value theorem there exits an $a> VaR_{\alpha}(X)$ such that
$(1+\theta)\mathbb{E}[ (X-a)_{+}] =P.$ Since $U^{-1}$ is increasing, problem
\eqref{eq:ORP} is equivalent to
\[
\label{eq:ORP2}\min\quad\mathbb{E} \Big[ U(R_{f}(X)) 1_{[R_{f}(X)\ge
VaR_{\alpha}(R_{f}(X))]} \big] \quad s.t.\quad(1+\theta) \mathbb{E}[f(X)]
=P,\; f \in\mathcal{C}.
\]
Since $f\in\mathcal{C}$ we have by Lemma \ref{lem:VarU} that $VaR_{\alpha
}(R_{f}(X))= R_{f}(VaR_{\alpha}(X))$ and since $R_{f}$ is non-decreasing we
obtain
\[
\{ X\ge VaR_{\alpha}(X)\} \subset\{ R_{f}(X)\ge R_{f}(VaR_{\alpha}(X)) =
VaR_{\alpha}(R_{f}(X))\}.
\]
On the other hand it is reasonable to assume that $f(x) = 0$ for $x\le
VaR_{\alpha}(X)$ since this probability mass does not enter the target
function which implies that $R_{f}(x) = x$ for $x\le VaR_{\alpha}(X)$ and
thus
\[
\{ R_{f}(X)\ge R_{f}(VaR_{\alpha}(X)) = VaR_{\alpha}(R_{f}(X))\} \subset\{
X\ge VaR_{\alpha}(X)\}.
\]
In total we have that
\[
\{ R_{f}(X)\ge VaR_{\alpha}(R_{f}(X))\} = \{ X\ge VaR_{\alpha}(X)\}.
\]
Hence, we can equivalently consider the problem
\[
\label{eq:ORP3}\min_f\quad\mathbb{E} \Big[ U(R_{f}(X)) 1_{[X\ge VaR_{\alpha
}(X)]} \big] \quad s.t.\quad(1+\theta) \mathbb{E}[f(X)] =P,\; f \in
\mathcal{C}.
\]

Next note that we have for any convex, differentiable function $g:
\mathbb{R}_{+}\to\mathbb{R}_{+}$ that
\[
g(x)-g(y) \ge g^{\prime}(y)(x-y), \quad x,y \ge0.
\]
Now consider the function $g(z) := U(x-z) 1_{[x\ge VaR_{\alpha}(X)]}+\lambda
z$ for fixed $\lambda:= U^{\prime}(a)>0$ and fixed $x\in\mathbb{R}_{+}$. By
our assumption $g$ is convex and differentiable with derivative $$g^{\prime}(z)
= -U^{\prime}(x-z) 1_{[x\ge VaR_{\alpha}(X)]} + \lambda.$$ Let $f^{*}$ be the
reinsurance treaty defined in the theorem and $f\in\mathcal{C}$ any other
admissible reinsurance treaty. Then
\begin{align*}
&  \mathbb{E}\big[ U(X-f(X)) 1_{[X\ge VaR_{\alpha}(X)]} - U(X-f^{*}(X))
1_{[X\ge VaR_{\alpha}(X)]} + \lambda(f(X)-f^{*}(X)) \big] \ge\\
&  \ge\mathbb{E}\Big[ \big( - U^{\prime}(X-f^*(X))1_{[X\ge VaR_{\alpha}(X)]} + \lambda\big)(f(X)-f^{*}(X)) \Big].
\end{align*}
Rearranging the terms and noting that $\mathbb{E}[f(X)]=\mathbb{E}[f^{*}(X)]$
we obtain
\begin{align*}
&  \mathbb{E}\big[ U(X-f(X)) 1_{[X\ge VaR_{\alpha}(X)]} \big] + \mathbb{E}
\Big[ \big( U^{\prime}(X-f^*(X))1_{[X\ge VaR_{\alpha}(X)]} - \lambda\big) (f(X)-f^{*}(X))\Big] \ge\\
&  \ge\mathbb{E}\big[ U(X-f^{*}(X)) 1_{[X\ge VaR_{\alpha}(X)]}\big]
\end{align*}
The statement follows when we can show that
\[
\mathbb{E}\Big[ \big( U^{\prime}(X-f^*(X))1_{[X\ge VaR_{\alpha}(X)]} - \lambda\big) (f(X)-f^{*}(X)) \Big]\le0.
\]
We can write
\begin{align*}
&  \mathbb{E}\Big[ \big( U^{\prime}(X-f^*(X))1_{[X\ge VaR_{\alpha}(X)]} - \lambda\big)
(f(X)-f^{*}(X)) \Big]\\
&  = \mathbb{E}\big[ 1_{[X \ge a ]}\big( U^{\prime}(X-f^*(X))1_{[X\ge VaR_{\alpha}(X)]} - \lambda\big)
(f(X)-f^{*}(X))  \big]+\\
&  + \mathbb{E}\big[ 1_{[X < a ]} \big( U^{\prime}(X-f^*(X))1_{[X\ge VaR_{\alpha}(X)]} - \lambda\big)
(f(X)-f^{*}(X))  \big]
\end{align*}
In the first case we obtain for $X\ge a$ by definition of $f^*$ and $\lambda$ (note that $a> VaR_\alpha(X)$):
\[
U^{\prime}(X-f^*(X))1_{[X\ge VaR_{\alpha}(X)]} - \lambda=U^{\prime}(a)-\lambda=0.
\]
In the second case we obtain for $X< a$ that $f(X)-f^{*}(X)=f(X)\ge0$ and since $U^{\prime}$ is increasing:
\[
U^{\prime}(X-f^*(X))1_{[X\ge VaR_{\alpha}(X)]} - \lambda\le\lambda1_{[X\ge
VaR_{\alpha}(X)]} - \lambda\le0.
\]
Hence the statement is shown.
\end{proof}

\section{ TQLM for symmetric loss models}
The symmetric family of distributions is well known to provide suitable distributions in finance and actuarial science. This family  generalizes the normal distribution into a framework of flexible distributions that are symmetric. We say that a real-valued random variable $X$ has a symmetric distribution, if its probability density function takes the form

\begin{equation}
f_{X}(x)=\frac{1}{\sigma}g\left(  \frac{1}{2}\left(  \frac{x-\mu}{\sigma
}\right)  ^{2}\right)  ,\quad x\in\mathbb{R} \label{S_pdf1}%
\end{equation}
where $g\left(  t\right)  \geq0,$ $t\geq0,$ is the density generator of $X$
and satisfies $$\int\limits_{0}^{\infty}t^{-1/2}g(t)dt<\infty.$$ The parameters
$\mu\in\mathbb{R}$ and\ $\sigma^{2}>0$ are the expectation and the scale
parameter of the distribution, respectively, and we write $X\backsim
S_{1}\left(  \mu,\sigma^{2},g\right)  $. If the variance of $X$ exists, then
it takes the form
\[
\mathbb{V}\left(  X\right)  =\sigma_{Z}^{2}\sigma^{2},
\]
where%
\[
\sigma_{Z}^{2}=2\underset{0}{\overset{\infty}{\int}}t^{2}g\left(  \frac{1}%
{2}t^{2}\right)  dt<\infty.
\]
For the sequel, we also define the standard symmetric random variable
$Z\backsim S_{1}\left(  0,1,g\right)  $ and a cumulative generator 
$\overline{G}(t),$ first defined in \citet{L1}, that takes the form $$\overline{G}
(t)=\underset{t}{\overset{\infty}{\int}}g(v)dv,$$ with the condition
$\overline{G}(0)<\infty.$ Special members of the family of symmetric distributions are:

\begin{enumerate}
\item[a)] The normal distribution, $g(u)=e^{-u}/\sqrt{2\pi},$
\item[b)] Student-t distribution $g(u)=\frac{\Gamma(\frac{m+1}{2})}%
{\Gamma(m/2)\sqrt{m\pi}}\left(  1+\frac{2u}{m}\right)  ^{-\left(  m+1\right)
/2}$ with $m>0$ degrees of freedom,
\item[c)] Logistic distribution, with $g\left(  u\right)  =ce^{-u}/\left(
1+e^{-u}\right)  ^{2}$ where $c>0$ is the normalizing constant.
\end{enumerate}

In what follows we will consider the TQLM for this class of random variables. 

\begin{theorem}\label{theo:TQLM_sym}
Let $X\backsim S_{1}(\mu,\sigma^{2},g)$. Then, the TQLM  takes the
following form
\begin{equation}
\rho_{U}^{\alpha}(X)=\rho_{\tilde{U}}^{\alpha}(Z)
\end{equation}
where $\tilde{U}(x)=U(\sigma x+\mu).$
\end{theorem}

\begin{proof}
For the symmetric distributed $X$, we have
\[
\rho_{U}^{\alpha}(X)= U^{-1}\left(  \frac{\mathbb{E}\big[U(X) 1_{\{X\geq
VaR_{\alpha}(X)\}} \big]}{1-\alpha}\right)  .
\]
Now we obtain
\begin{align*}
\mathbb{E}\big[U(X) 1_{\{X\geq VaR_{\alpha}(X)\}} \big]  &  = \int%
_{VaR_{\alpha}(X)}^{\infty}U(x) \frac{1}{\sigma} g\big(\frac12(\frac{x-\mu
}{\sigma})^{2}\big)dx\\
&  = \int_{\frac{VaR_{\alpha}(X)-\mu}{\sigma}}^{\infty}U(\sigma z+\mu)
g(\frac12 z^{2})dz = \int_{VaR_{\alpha}(Z)}^{\infty}\tilde{ U}(z) g(\frac12
z^{2})dz\\
&  = \mathbb{E}\big[\tilde{U}(Z) 1_{\{Z\geq VaR_{\alpha}(Z)\}} \big]
\end{align*}
where $\tilde{U}(x)=U(\sigma x+\mu).$ Hence the statement follows.
\end{proof}

For the special case of  Tail Conditional Entropic Risk Measures we obtain the following result:

\begin{theorem}
\label{theo:TCERMell} Let $X\backsim S_{1}(\mu,\sigma^{2},g)$. The moment
generating function of $X$ exists if and only if the Tail Conditional Entropic
Risk Measure satisfies
\[
\rho_{\gamma}^\alpha(X)=\mu+\sigma\rho_{\sigma\gamma}^\alpha(Z) <\infty.
\]
\end{theorem}

\begin{proof}
For a  function $U$ we obtain:
\begin{align*}
\mathbb{E}\Big[U(X)1_{[X\geq VaR_{\alpha}(X)]}\Big]  &  =\int_{VaR_{\alpha
}(X)}^{\infty}U(x)\frac{1}{\sigma}g\Big(\frac{1}{2}(\frac{x-\mu}{\sigma}
)^{2}\Big)dx\\
&  =\int_{\frac{VaR_{\alpha}(X)-\mu}{\sigma}}^{\infty}U(\sigma y+\mu
)g(\frac{1}{2}y^{2})dy.
\end{align*}
Plugging in $U(x)=\frac1\gamma e^{\gamma x}$ yields
\[
\mathbb{E}\Big[U(X)1_{[X\geq VaR_{\alpha}(X)]}\Big]=\frac1\gamma e^{\gamma\mu}\int%
_{\frac{VaR_{\alpha}(X)-\mu}{\sigma}}^{\infty}e^{\gamma\sigma y}g(\frac{1}
{2}y^{2})dy.
\]
Hence it follows that
\begin{align*}
\rho_{\gamma}^\alpha(X)  &  =\frac{1}{\gamma}\Big\{\gamma\mu+\log\Big(\int%
_{\frac{VaR_{\alpha}(X)-\mu}{\sigma}}^{\infty}e^{\gamma\sigma y}g(\frac{1}%
{2}y^{2})dy\Big)-\log(1-\alpha)\Big\}\\
&  =\mu+\sigma\frac{1}{\gamma\sigma}\log\Big(\int_{VaR_{\alpha}(Z)}^{\infty
}e^{\gamma\sigma y}g(\frac{1}{2}y^{2})dy\Big)+\sigma\frac{\log(1-\alpha
)}{\sigma\gamma}\\
&  =\mu+\sigma\rho_{\sigma\gamma}^\alpha(Z)
\end{align*}
Also note that $\rho_\gamma^\alpha(X)<\infty$ is equivalent to the existence of the moment generating function.
\end{proof}

\bigskip In the following theorem, we derive an explicit formula for the Tail
Conditional Entropic Risk Measure for the family of symmetric loss models.
For this, we denote the cumulant function of $Z$ by $\kappa\left(  t\right)
:=\log\psi\left(  -\frac{1}{2}t^{2}\right)$ where $\psi$ is the characteristic generator, i.e. it satisfies $\Eop[e^{itX}]= e^{it\mu} \psi(\frac12 t^2\sigma^2).$

\begin{theorem}\label{theo:TCERMexpl}
\bigskip Let $X\backsim S_{1}(\mu,\sigma^{2},g)$ and assume that the moment generating function of $X$ exists. Then the Tail Conditional Entropic Risk
Measure is given by
\[
\rho_{\gamma}^\alpha(X)=\mu+\gamma^{-1}\kappa\left(  \gamma\sigma\right)
+\log\left(  \frac{\overline{F}_{Y}\left(  VaR_{\alpha}\left(  Z\right)
\right)  }{1-\alpha}\right)  ^{-1/\gamma}.
\]
Here $F_{Y}\left(  y\right)  $ is the cumulative distribution function of a
random variable $Y$ with the density
\[
f_{Y}\left(  y\right)  =\frac{e^{\gamma\sigma y}}{\psi\left(  -\frac{1}
{2}\gamma^{2}\sigma^{2}\right)  }g\left(  \frac{1}{2}y^{2}\right)
,y\in\mathbb{R}
\]
and $\overline{F}_{Y}$ is its tail distribution function. 
\end{theorem}

\begin{proof}
From the previous Theorem, we have that $\rho_{\gamma}^\alpha(X)=\mu+\sigma
\rho_{\sigma\gamma}^\alpha(Z)$ where $Z\backsim S_{1}(0,1,g).$ Then, from \citet{L01}, the conditional characteristic function of the symmetric
distribution can be calculated explicitly, as follows:
\[
\mathbb{E}\left[  e^{\gamma\sigma Z}|Z\geq VaR_{\alpha}\left(  Z\right)
\right]  =\frac{\int\limits_{VaR_{\alpha}\left(  Z\right)  }^{\infty
}e^{\gamma\sigma z}g\left(  \frac{1}{2}z^{2}\right)  dz}{1-\alpha}%
\]
Observing that the following relation holds for any characteristic generator
$\psi$ of $g$ (see, for instance \citet{L01}, \citet{DHLx})
\[
\int\limits_{-\infty}^{a}e^{\gamma\sigma z}g\left(  \frac{1}{2}z^{2}\right)
dz=\psi\left(  -\frac{1}{2}\gamma^{2}\sigma^{2}\right)  F_{Y}\left(  a\right)
,a\in\mathbb{R},
\]
we conclude that
\[
\mathbb{E}\left[  e^{\gamma\sigma Z}|Z\geq VaR_{\alpha}\left(  Z\right)
\right]  =\psi\left(  -\frac{1}{2}\gamma^{2}\sigma^{2}\right)  \frac
{\overline{F}_{Y}\left(  VaR_{\alpha}\left(  Z\right)  \right)  }{1-\alpha},
\]
and finally,
\begin{align*}
\rho_{\gamma}^\alpha(X)  &  =\mu+\sigma\rho_{\sigma\gamma}^\alpha(Z)\\
&  =\mu+\gamma^{-1}\left[  \log\psi\left(  -\frac{1}{2}\gamma^{2}\sigma
^{2}\right)  +\log\frac{\overline{F}_{Y}\left(  VaR_{\alpha}\left(  Z\right)
\right)  }{1-\alpha}\right] \\
&  =\mu+\gamma^{-1}\kappa\left(  \gamma\sigma\right)  +\log\left(
\frac{\overline{F}_{Y}\left(  VaR_{\alpha}\left(  Z\right)  \right)
}{1-\alpha}\right)  ^{-1/\gamma}%
\end{align*}
where $\kappa\left(  \gamma\sigma\right)  =\log\psi\left(  -\frac{1}{2}%
\gamma^{2}\sigma^{2}\right)  $ is the cumulant of $Z.$
\end{proof}

\begin{example}
\bigskip Normal distribution. For $X\backsim N_{1}(\mu,\sigma^{2})$, the
characteristic generator is the exponential function, and we have%
\begin{equation}
\psi\left(  -\frac{1}{2}t^{2}\right)  =e^{\frac{1}{2}t^{2}}. \label{CGSN1t}%
\end{equation}
This leads to the following density of $Y$%
\begin{align}
f_{Y}(y)  &  =e^{\gamma\sigma y-\frac{1}{2}\gamma^{2}\sigma^{2}}\frac
{1}{\sqrt{2\pi}}e^{-\frac{1}{2}y^{2}}\label{pdf1t}\\
&  =\phi\left(  y-\gamma\sigma\right)  ,\nonumber
\end{align}
where $\phi$ is the standard normal density function. Then, the Tail
Conditional Entropic Measure is given by
\[
\rho_{\gamma}^\alpha(X)=\mu+\frac{\gamma}{2}\sigma^{2}+\log\left(  \frac
{\overline{\Phi}\left( \Phi^{-1}(\alpha)  -\gamma\sigma\right)
}{1-\alpha}\right)  ^{-1/\gamma}.
\]
Here $\Phi,\overline{\Phi}$ are the cumulative distribution function and the tail distribution function of the standard normal distribution, respectively.
\end{example}

\begin{remark}
The formulas of Theorem \ref{theo:TQLM_sym} and \ref{theo:TCERMexpl} can be specialized to recover existing formulas for the Conditional Tail Expectation, the Value at Risk and the Entropic Risk Measure of symmetric distributions. More precisely we obtain from Theorem \ref{theo:TCERMexpl} that
\begin{eqnarray*} CTE_\alpha(X) &=& \lim_{\gamma\downarrow 0} \rho_\gamma^\alpha(X) = \lim_{\gamma\downarrow 0}\Big[\mu+\gamma^{-1}\kappa\left(  \gamma\sigma\right)
+\log\left(  \frac{\overline{F}_{Y}\left(  VaR_{\alpha}\left(  Z\right)
\right)  }{1-\alpha}\right)  ^{-1/\gamma}\Big]\\
&=& \mu + \sigma \frac{\bar{G}(\frac12 VaR_\alpha(Z)^2)}{1-\alpha},
 \end{eqnarray*}
 where the first $\lim_{\gamma\downarrow 0} \gamma^{-1}\kappa\left(  \gamma\sigma\right)=0$ using L'Hopital's rule and the second limit is the stated expression by again using L'Hopital's rule. This formula can e.g. be found in  \citet{L01} Corollary 1. The Entropic Risk Measure can be obtained by
 $$  \lim_{\alpha\downarrow 0} \rho_\gamma^\alpha(X) = \mu-\gamma^{-1}\kappa\left(  \gamma\sigma\right)$$
 and for the Value at Risk we finally get with Theorem \ref{theo:TQLM_sym} and using 
 \[
U(x)=\left\{
\begin{array}
[c]{cl}%
x, & x\leq VaR_{\alpha}(X)\\
VaR_{\alpha}(X), & x>VaR_{\alpha}(X)
\end{array}
\right.
\]
that $$ VaR_\alpha(X) = \mu+\sigma VaR_\alpha(Z).$$
Thus our general formulas comprises several important special cases. 
 \end{remark}

\subsection{Optimal Portfolio Selection with Tail Conditional Entropic Risk
Measure}
The concept of optimal portfolio selection is dated back to  \citet{M3} and \citet{10dF40}, where the optimization of the mean-variance measure provides a portfolio selection rule that calculates the weights one should give to each investment of the portfolio in order to get the maximum return under a certain level of risk. In this Section, we examine the optimal portfolio selection with the TQLM measure for the multivariate elliptical models. The multivariate elliptical models of distributions are
defined as follows:

Let $\mathbf{X}$ be a random vector with values in $\mathbb{R}^{n}$
whose\ probability density function is given by (see for instance \citet{L1})
\begin{equation}
f_{\mathbf{X}}(\mathbf{x})=\frac{1}{\sqrt{|\Sigma|}}g_{n}\left(  \frac{1}%
{2}(\mathbf{x}-\mathbf{\mu})^{T}\Sigma^{-1}(\mathbf{x}-\mathbf{\mu})\right)
,\mathbf{x}\in\mathbb{R}^n. \label{Ellip1I1i}%
\end{equation}
Here $g_{n}\left(  u\right)  ,u\geq0,$ is the density generator of the
distribution that satisfies the inequality $$\int\limits_{0}^{\infty}%
t^{n/2-1}g_{n}(t)dt<\infty,$$ where $\mathbf{\mu}\in\mathbb{R}^{n}$ is the expectation of $\mathbf{X}$\ and $\Sigma$ is the $n\times n$
positive definite scale matrix, where, if exists, the covariance matrix of
$\mathbf{X}$ is given by
\[
Cov\left(  \mathbf{X}\right)  =\frac{\sigma_{Z}^{2}}{n}\Sigma,
\]
and we write $\mathbf{X}\backsim E_{n}(\mathbf{\mu},\Sigma,g_{n}).$ For $n=1$ we get the class of symmetric distributions discussed in the previous section. For a
large subset of the class of elliptical distributions, such as the normal,
Student-t, logistic, and Laplace distributions, for\ $\mathbf{X}\backsim
E_{n}(\mathbf{\mu},\Sigma,g_{n})$ and $\mathbf{\pi\in\mathbb{R}}^{n}$ be some non-random vector, we have that\ $\mathbf{\pi}^{T}%
\mathbf{X}\backsim E_{1}(\mathbf{\pi}^{T}\mathbf{\mu},\mathbf{\pi}^{T}
\Sigma\mathbf{\pi},g),$ $g:=g_{1}.$ This means that the linear transformation
of an elliptical random vector is also elliptically distributed with the same
generator $g_{n}$ reduced to one dimension. For instance, in the case of the
normal distribution $g_{n}\left(  u\right)  =e^{-u}/\left(  2\pi\right)
^{n/2},$ then $g\left(  u\right)  :=g_{1}\left(  u\right)  =e^{-u}/\left(
2\pi\right)  ^{1/2}.$

In modern portfolio theory, the portfolio return is denoted by $R:=\mathbf{\pi
}^{T}\mathbf{X}$ where it is often assumed that $\mathbf{X}\backsim
N_{n}(\mathbf{\mu},\Sigma)$ is a normally distributed random vector of financial returns.

\begin{theorem}
\label{theo:ETCERM} Let $\mathbf{X}\backsim E_{n}(\mathbf{\mu},\Sigma,g_{n}).$
Then,\ the Tail Conditional Entropic Risk Measure of the portfolio return
$R=\mathbf{\pi}^{T}\mathbf{X}$ is given by
\[
\rho_{\gamma}^\alpha(R)=\mathbf{\pi}^{T}\mathbf{\mu}+\sqrt{{\mathbf{\pi}^{T}
\Sigma\mathbf{\pi}}}\rho^\alpha_{\gamma\sqrt{\mathbf{\pi}^{T}
\Sigma\mathbf{\pi}}}(Z).
\]
\end{theorem}

\begin{proof}
From the linear transformation property of the elliptical random vectors, and
using Theorem \ref{theo:TCERMell}, the theorem immediately follows.
\end{proof}

Using the same notations and definitions as in \citet{L2}, we define a column
vector of $n$ ones, $\mathbf{1}$, and $\mathbf{1}_{1}$ as a column vector of
$(n-1)$ ones. Furthermore, we define the $n\times n$ positive definite scale
matrix $\Sigma$ with the following partition
\[
\Sigma=\left(
\begin{array}
[c]{cc}%
\Sigma_{11} & \mathbf{\sigma}_{1}\\
\mathbf{\sigma}_{1}^{T} & \sigma_{nn}%
\end{array}
\right)  .
\]
Here $\Sigma_{11}$ is an $\left(  n-1\right)  \times\left(  n-1\right)  $
matrix, $\mathbf{\sigma}_{1}=\left(  \sigma_{1n},...,\sigma_{n-1n}\right)
^{T}$ and $\sigma_{nn}$ is the $\left(  n,n\right)  $ component of $\Sigma,$
and we also define a $\left(  n-1\right)  \times\left(  n-1\right)  $ matrix
$Q,$
\[
Q=\Sigma_{11}-\mathbf{1}_{1}\sigma_{1}^{T}-\sigma_{1}\mathbf{1}_{1}^{T}%
+\sigma_{nn}\mathbf{1}_{1}\mathbf{1}_{1}^{T}%
\]
which is also positive definite (see again \citet{L2}). We also define the
$\left(  n-1\right)  \times1$ column vector
\[
\mathbf{\Delta}=\mu_{n}\mathbf{1}_{1}-\mathbf{\mu}_{1}%
\]
where $\mathbf{\mu}_{1}:=\left(  \mu_{1},\mu_{2},...,\mu_{n-1}\right) ^{T}.$
In what follows we consider the problem of finding the portfolio with the
least $\rho_{\gamma}^{\alpha}$ for fixed $\alpha$ and $\gamma$:
\begin{equation}
\min_\pi \quad\rho_{\gamma}^{\alpha}(R)\quad s.t.\quad\sum_{i=1}^{n}\pi_{i}=1.
\label{eq:OPP}%
\end{equation}
The solution is given in the next theorem:

\begin{theorem}
Let $\mathbf{X}\backsim E_{n}(\mathbf{\mu},\Sigma,g_{n})$ be a random vector
of returns, and let $R=\mathbf{\pi}^{T}\mathbf{X}$ be a portfolio return of
investments $X_{1},X_{2},...,X_{n}.$ Then, the optimal solution to
\eqref{eq:OPP} is
\[
\mathbf{\pi}^{\ast}=\mathbf{\varphi}_{1}+r^{\ast}\mathbf{\varphi}_{2}%
\]
if%
\[
r\cdot s_{1}\left(  \mathbf{\Delta}^{T}Q^{-1}\mathbf{\Delta\cdot}r^{2}+\left(
\mathbf{1}^{T}\Sigma^{-1}\mathbf{1}\right)  ^{-1}\right)  =1/2
\]
has a unique positive solution $r^{\ast}$. Here
\[
\mathbf{\varphi}_{1}\mathbf{=}\left(  \mathbf{1}^{T}\Sigma^{-1}\mathbf{1}%
\right)  ^{-1}\Sigma^{-1}\mathbf{1,}%
\]%
\[
\mathbf{\varphi}_{2}=\left(  \mathbf{\Delta}^{T}Q^{-1},-\mathbf{1}_{1}%
^{T}Q^{-1}\mathbf{\Delta}\right)  ^{T},
\]
and $s_{1}=ds(t)/dt,$ $s\left(  t\right)  =t^{2}\rho_{t^{2}\gamma}^\alpha(Z).$
\end{theorem}

\begin{proof}
We first observe by Theorem \ref{theo:ETCERM} that the minimization of
$\rho_{\gamma}^\alpha(R)$ is achieved when minimizing $\mathbf{\pi}^{T}\mathbf{\mu}+\sqrt{\mathbf{\pi}^{T}\Sigma\mathbf{\pi}}\rho^\alpha_{\gamma\sqrt{\mathbf{\pi}^{T}
\Sigma\mathbf{\pi}}}(Z).$ Then, using Theorem 3.1 in \citet{L2} (see also \citet{L4}) the statement immediately follows.
\end{proof}

\section{Discussion}
The Tail Quasi-Linear Mean is a measure which focuses on the right tail of a risk distribution. In its general definition it comprises a number of well-known risk measures like Value at Risk, Conditional Tail Expectation and Entropic Risk Measure. Thus, once having results about the TQLM we are able to specialize them to other interesting cases. It is also in line with the actuarial concept of a Mean Value principle. Moreover, we have shown that it is indeed possible to apply the TQLM in risk management and that it yields computationally tractable results.\\

{\bf Acknowledgements:} This research was supported by the Israel Science Foundation (Grant No. 1686/17 to T.S.)

\end{document}